\documentclass[12pt]{amsart}
\usepackage{amsmath,amscd,amsfonts,amssymb, color,bbm, bm}
\usepackage{latexsym, graphicx, pstricks,rotating, enumerate}
\usepackage[pdftex,bookmarks,colorlinks]{hyperref}
\definecolor{darkblue}{rgb}{0.0,0.0,0.3}
\hypersetup{colorlinks=true,citecolor=darkblue,
unicode=true,pdftitle=toeplitz.pdf,pdfauthor=Ritabrata
Sengupta,bookmarks=false, urlcolor=darkblue} 
\numberwithin{equation}{section}
\newtheorem{prop}{Proposition}[section]
\newtheorem{defin}{Definition}[section]
\newtheorem{rem}{Remark}[section]

\newtheorem{theorem}{Theorem}[section]
\newtheorem{lemma}{Lemma}[section]
\newtheorem{corollary}{Corollary}[section]

\newcommand{\tr}{\mathrm{Tr}}
\newcommand{\dd}{\mathrm{d}}

\newcommand{\ket}[1]{\ensuremath{\left|#1\right\rangle}}
\newcommand{\braket}[2]{\ensuremath{\left|#1\right\rangle
\hspace{-3pt}\left\langle #2\right|}}
\begin{document}
\title{Exchangeable, stationary and entangled chains of
Gaussian states}
\author{K. R. Parthasarathy}
\address{Theoretical Statistics and Mathematics Unit, Indian
Statistical Institute, Delhi Centre, 7 S J S Sansanwal Marg,
New Delhi 110 016, India}
\email[K R Parthasarathy]{\href{mailto:krp@isid.ac.in}{krp@isid.ac.in}}
\author{Ritabrata Sengupta}
\address{Theoretical Statistics and Mathematics Unit, Indian
Statistical Institute, Delhi Centre, 7 S J S Sansanwal Marg,
New Delhi 110 016, India}
\email[Ritabrata Sengupta]{\href{mailto:rb@isid.ac.in}{rb@isid.ac.in}}
\begin{abstract}
\par We explore conditions on the covariance matrices of a
consistent chain of mean zero finite mode Gaussian states in
order that the chain may be exchangeable or stationary. For
an exchangeable chain our conditions are necessary and
sufficient. Every stationary Gaussian chain admits an
asymptotic entropy rate. Whereas an exchangeable chain
admits a simple expression for its entropy rate, in our
examples of stationary chains the same admits an integral
formula based on the asymptotic eigenvalue distribution for
Toeplitz matrices. An example of a stationary entangled
Gaussian chain is given.

\smallskip
\noindent \textbf{Keywords.} Gaussian state, exchangeable,
stationary and entangled Gaussian chains, entropy rate. 

\smallskip

\noindent \textbf{Mathematics Subject Classification
(2010):} 81P45, 94A15, 54C70.
\end{abstract}
\thanks{The authors thank Prof. Ajit Iqbal Singh for 
useful comments and suggestions. RS acknowledges financial
support from the National Board for Higher Mathematics,
Govt. of India.} 
\maketitle
\section{Introduction}

\par The importance of finite mode Gaussian states and their
covariance matrices in general quantum theory as well as
quantum information has been highlighted extensively in the
literature. A comprehensive survey of Gaussian states and
their properties can be found in the
book of Holevo \cite{MR2797301}. For their applications to 
quantum information theory the reader is referred to the
survey article by Weedbrook et al \cite{RevModPhys.84.621}
as well as Holevo's book \cite{zbMATH06069722}.  For our
reference we use \cite{arvindsurvey, MR2662722,
zbMATH06185273} for Gaussian states and for notations in the
following sections we use \cite{krprb}. While working
on this paper, Mozrzymas, Rutkowski, and Studzi{\'n}ski had
posted in arXiv \cite{1505.06422} an article similar in
spirit, but for finite dimensional Hilbert spaces.

\par In the present paper our concern is with a chain of
finite mode Gaussian states constituting a consistent
sequence exhibiting properties like exchangeability,
stationarity, and entanglement. All these notions are easily
translated into properties of infinite covariance matrices. 

\par If $\rho$ is a state of a quantum system and
$X_i,\,i=1,2$ are two real valued observables, or
equivalently, self-adjoint operators with finite second
moments in a state $\rho$ then the covariance between $X_1$
and $X_2$ in the state $\rho$ is the scalar quantity
$\tr\left(\frac{X_1 X_2 +X_2 X_1}{2}\right)\rho - \tr X_1
\rho\, \cdot\, \tr X_2\rho$, which is denoted by
$\mathrm{Cov}_\rho (X_1,X_2)$. Suppose $q_1,p_1;\,
q_2,p_2;\, \cdots ;\, q_n,p_n$ are the position - momentum
pairs of observables of a quantum system with $n$ degrees of
freedom obeying the canonical  commutation relations. Then
we express 
\[(X_1,X_2,\cdots,X_{2n}) =
(q_1,p_1,q_2,p_2,\cdots,q_n,p_n).\]
If $\rho$ is a state in which all the $X_j$'s have finite
second moments we write 
\begin{equation}\label{eq:1.1}
S_\rho = [[ \mathrm{Cov}_\rho(X_i, X_j)]], \quad i,j \in \{
1,2,\cdots, 2n\}.
\end{equation} 
We call $S_\rho$ the covariance matrix of the position
momentum observables. If we write 
\begin{equation}\label{eq:1.2}
J_{2n}= \begin{bmatrix}\begin{array}{rr}
0 & 1 \\ -1 & 0
\end{array} &&& \\
& \begin{array}{rr}
0 & 1 \\ -1 & 0
\end{array} &&\\ 
 &  & \ddots & \\
&&& \begin{array}{rr}
0 & 1 \\ -1 & 0
\end{array}
\end{bmatrix}
\end{equation}
or equivalently $\bigoplus_1^n \begin{bmatrix} 0 & 1 \\ -1 &
0 \end{bmatrix} $ for the $2n \times 2n$ block diagonal
matrix, the complete Heisenberg uncertainty principle for
all the position and momentum observables assumes the form
of the following matrix inequality 
\begin{equation}\label{eq:1.3}
S_\rho +  \frac{\imath}{2} J_{2n} \geq 0.
\end{equation} 
Conversely, if $S$ is any real $2n \times 2n$ symmetric
matrix obeying the inequality $S_\rho +  \frac{\imath}{2}
J_{2n} \geq 0$, then there exists a state $\rho$ such that
$S$ is the covariance matrix $S_\rho$ of the observables $q_1,p_1;\,
q_2,p_2;\, \cdots ;\, q_n,p_n$. In such a case $\rho$ can be
chosen to be a Gaussian state with mean zero. In view of
this we make a formal definition.
\begin{defin}\label{def:1.1}
A $2n \times 2n$ real symmetric positive matrix $S$ is said
to be a \emph{G-matrix} if it satisfies the inequality 
\begin{equation} \label{eq:1.4}
S +  \frac{\imath}{2} J_{2n} \geq 0.
\end{equation}
\end{defin}

\par Suppose 
\begin{equation}\label{eq:1.5}
\Sigma = [[A_{ij}]], \quad i,\, j \in \{ 1,2, \cdots\}
\end{equation} 
is an infinite matrix where each $A_{ij}$ is a $2k \times
2k$ real matrix and $A_{ij}^T=A_{ji}$ for all $i,\,j$.  For
any finite subset $I=\{i_1 < i_2 < \cdots < i_n\} \subset \{
1, 2, \cdots \}$, let  
\begin{equation}\label{eq:1.6}
\Sigma(I) =[[A_{i_r i_s}]],\quad r,\,s \in \{1,2,\cdots\}
\end{equation}
be the $2kn \times 2kn$ matrix obtained from $\Sigma$ by
restriction to its rows and columns numbered $  i_1 < i_2 <
\cdots < i_n$. 

\begin{defin}\label{def:1.2}
We say that $\Sigma$ is a \emph{G-chain of
order $k$}  if $\Sigma(I)$ is the covariance matrix of a
$kn$-mode zero mean Gaussian state $\rho(I)$ in the boson
Fock space $\Gamma(\mathbb{C}^{kn}) = \mathcal{H}_{i_1}
\otimes \mathcal{H}_{i_2} \otimes  \cdots \otimes
\mathcal{H}_{i_n}$ where $\mathcal{H}_j$ denotes the $j$-th
copy of the Hilbert space $\mathcal{H} =
\Gamma(\mathbb{C}^k)$, $j=1,2, \cdots.$ 
\end{defin}

\par If $I= \{ i_1 < i_2 < \cdots < i_n \}$ and $I'= \{ i_1 < i_2
< \cdots < i_{n+m} \} $ then clearly $\rho(I)$ is the
$I$-marginal of the state $\rho(I')$. Thus $\Sigma$ describes
a consistent family of zero mean Gaussian states with finite
number of modes which are multiples of $k$. It is a standard
result of Gaussian states that $\Sigma$ is a G-chain of
order $k$ if and only if the following matrix inequalities
hold:
\begin{equation}\label{eq:1.7}
\Sigma(\{1,2,\cdots, n\}) +\frac{\imath}{2} J_{2kn} \ge 0, \quad n=1,2,\cdots
\end{equation}

\begin{defin}\label{def:1.3} 
We say that $\Sigma$ in equation (\ref{eq:1.5}) is an
\emph{exchangeable  G-chain}  if it is a G-chain and there
exist two $2k \times 2k$ matrices $A,~B$ such that
\begin{equation}\label{eq:1.8}
A_{ij} = \left\{
\begin{array}{lcl}
B & \text{if} & j>i, \\
A & \text{if} & j=i, \\
B^T & \text{if} & j<i. 
\end{array} \right.
\end{equation}
In such a case we write
\begin{eqnarray}\label{eq:1.9}
\Sigma &=& \Sigma(A,B),\nonumber\\
\Sigma(I) &=& \Sigma(I; A,B).
\end{eqnarray}
\end{defin}

\par It is clear that for any two finite subsets $I$ and $I'$
of $\{1,2,\cdots \}$ with the same cardinality an
exchangeable G-chain $\Sigma$ satisfies 
\begin{equation*}
\Sigma(I;A,B) = \Sigma(I';A,B)
\end{equation*}
and therefore the corresponding Gaussian states $\rho(I)$
and $\rho(I')$ are same, i.e. the quantum version of de Finitti
exchangeability property \cite{deFine} holds.  

\par In this paper we shall show that two $2k \times 2k$
matrices $A$ and $B$ determine a G-chain $\Sigma(A,B)$ if
and only if $B=B^T,~B>0$ and $A-B$ is the covariance matrix
of a $k$-mode Gaussian state.

\par For any finite mode Gaussian state $\rho$ with
covariance matrix $C$ denote its von Neumann entropy
$S(\rho)$ by $S(C)$. we shall prove that for any
exchangeable G-chain $\Sigma(A,B)$, the sequence
$\left\{\frac{1}{n} S(\Sigma(\{1,2,\cdots,n\},A,B))\right\}$ decreases
monotonically to the limit $S(A-B)$. 

\begin{defin}\label{def:1.4}
We say that a G-chain $\Sigma$ given by (\ref{eq:1.5}) is
\emph{stationary} if there exist $2k \times 2k$ matrices
$A,\, B_1,\, B_2,\cdots$ such that
\begin{equation}\label{eq:1.10}
A_{ij} = \left\{
\begin{array}{lcl}
A & \text{if} & i=j, \\
B_{j-i} & \text{if} & j>i, \\
B_{i-j}^T & \text{if} & j<i
\end{array} \right.
\end{equation}
for all $i,\,j$.
\end{defin} 

\par For such a stationary G-chain the Gaussian states
$\rho(I)$ and $\rho(I+1)$ are same for any finite subset
$I\subset \{1,2,\cdots\}$. This translation invariance
property shows that our definition of stationarity is
similar to such a notion in classical theory of the
stochastic processes \cite{MR0058896}. It is an interesting
problem to find necessary and sufficient conditions on the
matrices $A,\, B_1, \, B_2,\cdots $ so that (\ref{eq:1.10})
yields a stationary G-chain. Even though we do not have an
answer to this question  we shall construct a large class of
examples of such stationary G-chains. 

\par For any stationary G-chain $\Sigma$, consider the von
Neumann entropy of the Gaussian state
$\rho(\{1,2,\cdots,n\},A,B)$, namely,
$S(\Sigma(\{1,2,\cdots,n\}:A,B))$. The strong sub-additivity
property of entropy \cite{MR0345558} implies  that the sequence
$\left\{\frac{1}{n} S(\Sigma(\{1,2,\cdots,n\}: A,B))
\right\}$ decreases monotonically to a limit $\bar{S}(\Sigma)$,
which may be called the \emph{entropy rate} of the stationary
chain. Using the  Kac-Murdock-Szeg{\"o} theory
for asymptotic eigenvalue distributions of
Toeplitz matrices \cite{MR0059482, MR761763} we shall get an
integral formula for the entropy rate of an interesting
example of a stationary G-chain of order $k$. 

\par When $k=1$, we shall construct a stationary G-chain in
which states of the form $\rho(\{1,2\})$ and hence $\rho( \{
1,2,\cdots, n\}),~n=2,3,4,\cdots$ are entangled. However it
would be more interesting to find examples of stationary
G-chains in which, for some fixed finite nonempty subset $I$
of $\{1,2,\cdots\}$, states of the form $\rho(I\cup
(I+n))$ with marginals $\rho(I)$ remain entangled for
arbitrarily large $n$. That would ensure preservation of
entanglement after an arbitrarily large lapse of time. 


\section{Exchangeable G-chain}

\par Our first main result gives necessary and sufficient
conditions on a pair $(A,B)$ of real $2k \times 2k$ matrices
so that the infinite matrix $\Sigma(A,B)$ in Definition
\ref{def:1.3} is an exchangeable G-chain.

\begin{theorem}\label{th:2.1}
Let $(A,B)$ be a pair of real $2k \times 2k$ matrices. Then
$\Sigma(A,B)$ is a G-chain if and only if $A$ and $B$ are
nonnegative definite and $A-B$ is a G-matrix.
\end{theorem}

\begin{proof}
Choose and fix a positive integer $n$ and put $I= \{ 1,2,
\cdots, n\}$. Write 
\[\Sigma_n(A,B) = \Sigma (I:A,B).\]
Denote by $I_n$ the identity matrix of order $n$ and $N_n$
the upper triangular nilpotent matrix with all its upper
triangular entries equal to $1$ and the rest equal to zero.
Let 
\[\ket{\psi_n} =\frac{1}{\sqrt{n}}[1,1,\cdots,1]^T,\]
a unit column vector of length $n$. Then 
\[I_n + N_n + N_n^T = n \braket{\psi_n}{\psi_n}.\]
Then we have 
\[\Sigma_n(A,B) = A \otimes I_n +B\otimes N_n + B^T \otimes
N_n^T.\]
Noting that $J_{2kn}= J_{2k}\otimes I_n$ we can now write
\begin{eqnarray}
\Sigma_n(A,B) +\frac{\imath}{2} J_{2kn} &=& \left( A +
\frac{\imath}{2} J_{2k} - \frac{1}{2} (B + B^T) \right)
\otimes \left( I_n - \braket{\psi_n}{\psi_n} \right)
\nonumber\\
&& +  \left( A + \frac{\imath}{2} J_{2k} + \frac{1}{2} (n-1)
(B + B^T) \right) \otimes \braket{\psi_n}{\psi_n} \nonumber
\\
&& + \frac{1}{2}(B-B^T) \otimes (N_n -N_n^T). \label{eq:2.1}
\end{eqnarray}
In order that $\Sigma(A,B)$ may be a G-chain it is necessary
and sufficient that the left hand side of (\ref{eq:2.1}) is
nonnegative definite for every $n$. To prove necessity we
multiply first both sides of (\ref{eq:2.1})  by $I_n \otimes
\braket{\psi_n}{\psi_n}$ and take relative trace over the
second component. Noting that for any real vector
$\ket{\psi}$ in the second Hilbert space, $\langle \psi| N_n -
N_n^T | \psi \rangle =0$ we get the inequality
\begin{equation}\label{eq:2.2}
 A + \frac{\imath}{2} J_{2k} + \frac{1}{2} (n-1) (B + B^T)
\geq 0 \quad \forall n.
\end{equation}
Dividing by $n-1$ and letting $n \rightarrow \infty $ we get
\begin{equation}\label{eq:2.3}
\frac{1}{2} (B + B^T) \geq 0.
\end{equation}
Choosing an arbitrary real unit vector $\ket{\psi}$ in the
range of $I_n - \braket{\psi_n}{\psi_n}$, multiplying
both sides of (\ref{eq:2.1}) by $I \otimes
\braket{\psi}{\psi}$ and tracing out over the second Hilbert
space we get 
\begin{equation}\label{eq:2.4}
A+ \frac{\imath}{2} J_{2k} - \frac{1}{2} (B + B^T) \geq 0.
\end{equation}
In other words $A - \frac{1}{2} (B + B^T)$ is a G-matrix.
Now we consider the complex unit vector 
\[ \ket{\phi_n} = \frac{1}{\sqrt{n}}
[1,\omega,\omega^2,\cdots, \omega^{n-1}]^T \]
where $\omega=e^{\frac{2\pi\imath}{n}}$, an $n$-th root of
unity. Simple algebra shows that 
\[ \langle \phi_n | N_n | \phi_n \rangle = \frac{1}{
\bar{\omega}-1} \]
and therefore 
\[\langle \phi_n | N_n - N_n^T | \phi_n \rangle =2\imath Im
\frac{1}{ \bar{\omega}-1} = \imath \cot \frac{\pi}{n}. \]
Now, multiplying both sides of (\ref{eq:2.1}) by $I_n
\otimes \braket{\phi_n}{\phi_n}$ and tracing over the second
Hilbert space we get the inequality 
\begin{equation}\label{eq:2.5}
A+ \frac{\imath}{2} J_{2k} - \frac{1}{2} (B + B^T) + \imath
\frac{B-B^T}{2} \cot \frac{\pi}{n} \geq 0
\end{equation}
for $n=1,2,\cdots$. Multiplying by $\tan \frac{\pi}{n}$ (for
$n \ge 3$) and letting $n \rightarrow \infty$ we get the
inequality 
\[\frac{\imath}{2}(B-B^T)\ge 0.\]
Since the left hand side is a Hermitian matrix with trace zero it
follows that $B=B^T$ and (\ref{eq:2.4}) implies that $A-B$
is a G-matrix. Together with (\ref{eq:2.3}) the proof of
necessity is complete.

\par To prove sufficiency, observe that $A-B$ being a
G-matrix and $B$ being positive implies that $A-B +n B$ is a
G-matrix for every $n=1,2,\cdots$. Now the identity
(\ref{eq:2.1}) implies that 
\begin{eqnarray*}
\Sigma_n(A,B) + \frac{\imath}{2} J_{2kn} &=& \left[(A-B) +
\frac{\imath}{2} J_{2k} \right] \otimes \left( I_n -
\braket{\psi_n}{\psi_n} \right)  + \left[ A+(n-1)B + \frac{\imath}{2} J_{2k} \right] \otimes
\braket{\psi_n}{\psi_n} \ge0
\end{eqnarray*}
for every $n$ and therefore $\Sigma_n(A,B)$ is a G-matrix for
every $n$. In other words $\Sigma(A,B)$ is a G-chain.
\end{proof}

\begin{corollary}\label{cor:2.1}
In any exchangeable G-chain $\Sigma(A,B)$  of order $k$, for
every finite set $I \subset \{1,2,\cdots\}$, the underlying
Gaussian state $\rho(I)$ is separable.
\end{corollary}
\begin{proof}
Without loss of generality we may assume that
$I=\{1,2,\cdots,n\}$  for some $n$. Then the covariance
matrix of $\rho(I)$ is equal to the $n \times n$ block
matrix 
\[ \begin{bmatrix}
A & B & \cdots & B \\
B & A & \cdots & B \\
\vdots &\vdots &\ddots &\vdots \\
B & B & \cdots & A
\end{bmatrix} = (A-B) \otimes I_n + B \otimes \begin{bmatrix}
1 & 1 & \cdots & 1 \\
1 & 1 & \cdots & 1 \\
\vdots &\vdots &\ddots &\vdots \\
1 & 1 & \cdots & 1
\end{bmatrix}.
\]
By Theorem \ref{th:2.1}, $(A-B)\otimes I_n$ is the
covariance matrix of an $n$-fold product Gaussian state and
the second summand on the right hand side of the equation
above is a nonnegative definite matrix. Hence by Werner and
Wolf's theorem \cite{WernerWolf} $\rho(I)$ is separable. 
\end{proof}

\section{Examples of stationary G-chain}
\par Let $A,\, B$ be real $2k \times 2k$ symmetric matrices.
For any fixed $j=1,2,\cdots$, denote by $\Delta^j(A,B)$ the
infinite block matrix all of whose diagonal blocks are equal
to $A$, $(n,n+j)$-th and $(n+j,n)$-th blocks are equal to
$B$ for every $n$ and all the remaining blocks are zero
matrices of order  $2k \times 2k$. For example, 
\[
\Delta^1 (A,B) = \begin{bmatrix}
A & B & 0 & 0 & 0 & \cdots \\
B & A & B & 0 & 0 & \cdots \\
0 & B & A & B & 0 & \cdots \\
\vdots & \vdots & \vdots & \vdots & \vdots & \ddots 
\end{bmatrix}.\] 
Denote by $\Delta_n^j(A,B)$ the $2kn \times 2kn $ matrix
obtained by $\Delta^j(A,B)$ by restriction to the first $n$
row and column blocks. For example, 
\begin{eqnarray*}
\Delta_2^1(A,B) &=& \begin{bmatrix} A & B \\ B & A
\end{bmatrix} \\
\Delta_3^2(A,B) &=& \begin{bmatrix} A &0 & B \\ 0 & A & 0 \\
B & 0 & A
\end{bmatrix} 
\end{eqnarray*}
and so on. 
\par Our first result gives a necessary and sufficient
condition for  $\Delta^j(A,B)$ to be a G-chain. 

\begin{theorem}\label{th:3.1}
Let $A,\,B$ be a pair of $2k \times 2k$ real symmetric
matrices. In order that $\Delta^j(A,B)$ may be a G-chain of
order $k$ it is necessary and sufficient that $A+tB$ is a
G-matrix for every $t\in[-2,2]$.
\end{theorem}

\begin{proof}
Denote by $L_n^j$ the upper triangular matrix whose
$(j+1)$-th upper diagonal entries are all equal to $1$ and
all the remaining entries are zero. Thus $L_n^j$ is defined
for $1 \leq j \leq n-1$. Then 
\begin{equation}\label{eq:3.1}
\Delta_n^j(A,B) = A \otimes I_n  + B \otimes (L_n^j +
(L_n^j)^T).
\end{equation}
Consider the spectral decomposition of the $n \times n$
symmetric matrix $L_n^j + (L_n^j)^T$:
\begin{equation}\label{eq:3.2}
L_n^j + (L_n^j)^T = \sum_{r=1}^n \lambda_{nr}
\braket{\psi_{nr}}{\psi_{nr}}
\end{equation}
where $\{\lambda_{nr}:\, r=1,2,\cdots,n\}$ are the
eigenvalues and $\{\ket{\psi_{nr}}:\, r=1,2,\cdots,n\}$ are
the corresponding orthonormal basis of eigenvectors for
$L_n^j + (L_n^j)^T$. Since each $L_n^j$ is a matrix with
operator norm equal to unity and therefore $L_n^j +
(L_n^j)^T$ has operator norm not exceeding $2$ it is clear
that 
\begin{equation}\label{eq:3.3}
|\lambda_{nr}| \le 2, \quad 1\leq r \leq n, ~~ n=1,2,\cdots.
\end{equation}
Equations (\ref{eq:3.1})--(\ref{eq:3.2}) imply 
\begin{equation}\label{eq:3.4}
\Delta_n^j(A,B) + \frac{\imath}{2} J_{2kn} = \sum_{r=1}^n
\left ( A + \lambda_{nr} B  + \frac{\imath}{2} J_{2k}
\right) \otimes \braket{\psi_{nr}}{\psi_{nr}}.
\end{equation}
Thus $\Delta_n^j(A,B)$ is a G-matrix if and only if $A+
\lambda_{nr}B $ is a G-matrix for each $r=1,2,\cdots,n$.
This together with (\ref{eq:3.3}) already proves the
sufficiency part of the theorem.

\par To prove necessity, we appeal to the theorem of  Kac,
Murdock and Szeg\"o \cite{MR0059482}. Consider the
probability distribution 
\[\mu_n = \frac{1}{n} \sum_{r=1}^n \delta_{\lambda_{nr}}\]
where $\lambda_{nr},\, r=1,2,\cdots,n$ are as in 
(\ref{eq:3.2}). The left hand side of (\ref{eq:3.2})
 is a Toeplitz matrix of order $n$ for each
$n$. Thus Kac, Murdock, Szeg\"o theorem implies that the
sequence $\{\mu_n\}$ converges weakly as $n \rightarrow
\infty$ to the probability measure $Lh^{-1}$ where $L$ is
the Lebesgue measure in the unit interval and $h(t) =
2 \cos2 \pi j t, \, t\in[0,1]$. 
This, in particular, implies that
$\{ \lambda_{nr}:\, r =1 , 2, \cdots, n,\, n=1,2,\cdots\}$
is dense in the interval $[-2,2]$. The proof of necessity is
now complete. 
\end{proof}

\begin{corollary}\label{cor:3.2}
Let $A,\,B_1,\,B_2,\,\cdots $ be real $2k \times 2k$
symmetric matrices satisfying the condition that $A+tB_j$ is
a G-matrix for every $j=1,2,\cdots$ and $t\in [-2,2]$.
Suppose $p_1,p_2,\cdots, $ is a probability distribution on
the set $\{1,2,\cdots\}$. Then the block Toeplitz matrix 
\[ \Sigma(A;p_1B_1,p_2B_2,\cdots) = \begin{bmatrix}
A & p_1B_1 & p_2B_2& \cdots & \cdots \\
p_1B_1 & A & p_1B_1 & p_2B_2 & \cdots \\
p_2B_2 & p_1B_1 & A & p_1B_1 & \cdots \\
\vdots &\vdots &\vdots  & \ddots & \ddots
\end{bmatrix}\]
is a stationary G-chain
\end{corollary}
\begin{proof}
This is immediate from the theorem because 
\[\Sigma(A;p_1B_1,p_2B_2,\cdots) = \sum_{j=1}^\infty p_j
\Delta^j (A,B_j)\]
and each $\Delta^j (A,B_j)$ is a G-chain.
\end{proof}

\section{Entropy rate of the stationary G-chain}

\par Suppose $\Sigma =\Sigma  (A,B_1,B_2,\cdots)$ is a
stationary G-chain. For any G-matrix $C$ denote by $S(C)$
the von Neumann entropy of a Gaussian state $\rho$ with
covariance matrix $C$. Let 
\begin{eqnarray*}
\Sigma_n &=& \Sigma(\{1,2,\cdots,n\}),\\
S_n &=& S(\Sigma_n).
\end{eqnarray*}
\begin{prop}\label{prop:4.1}
The sequences $\{S_n -S_{n-1}\}, ~ \left\{ \frac{1}{n} S_n
\right\}$ monotonically decrease to the same limit
$\bar{S}\ge 0$ as $n \rightarrow \infty$. Furthermore, $S_n
\ge S_{n-1}$ for all $n$.
\end{prop}

\begin{proof}
Consider three Gaussian quantum systems $P,\,Q,\,R$ so that
$PQR$ is in the mean zero Gaussian state $\rho(\{1,2,\cdots,
n+1\})$, $Q$ in $\rho(\{2,\cdots, n\})$, $PQ$ in
$\rho(\{1,2,\cdots, n\})$ and $QR$ in $\rho(\{2,\cdots,
n+1\})$. Then using stationarity we have 
\begin{eqnarray*}
S(\rho(PQR)) &=& S_{n+1} \\
S(\rho(PQ)) &=& S_{n} \\
S(\rho(QR)) &=& S_{n} \\
S(\rho(Q)) &=& S_{n-1}.
\end{eqnarray*}
By the strong subadditivity property of entropy
\cite{MR0345558} we have 
\[S_{n+1} + S_{n-1} \le 2 S_n\]
or 
\[S_{n+1} - S_n \le S_n - S_{n-1}.\]
Since 
\[\frac{S_n}{n} =\frac{(S_n - S_{n-1})+(S_{n-1} - S_{n-2}) +
\cdots + (S_1 - S_0)}{n}\]
where $S_0$ is defined to be zero, it follows that $\frac{S_n}{n}$
decreases monotonically to a limit $\bar{S}\ge0$. This also
implies that $S_n - S_{n-1}$ cannot decrease to $-\infty$
and hence $S_n - S_{n-1}$ also decreases monotonically to
$\bar{S}$. This also shows that $S_n \ge S_{n-1}$ for all
$n$. 
\end{proof} 
\begin{defin}\label{def:4.1}
We denote the limit $\bar{S}$ in Proposition
\ref{prop:4.1} by $\bar{S}(\Sigma)$ and call it the 
\emph{entropy rate} of the stationary G-chain $\Sigma$.
\end{defin}

\begin{theorem}\label{th:4.1}
Let $\Sigma =\Sigma(A,B)$ be an exchangeable G-chain. Then 
\[\bar{S}(\Sigma) = S(A-B).\]
\end{theorem}
\begin{proof}
If $C,\,D$ are two G-matrices so is $C\oplus D$ and $S(C
\oplus D) =S(C) +S(D)$. By Theorem \ref{th:2.1} and equation
(\ref{eq:2.1}) we have the identity
\[\Sigma_n(A,B)=(A-B) \otimes (I_n - \braket{\psi_n}{\psi_n})
+(A+\overline{n-1}B) \otimes \braket{\psi_n}{\psi_n}\]
where we have adopted the notations in the proof of Theorem
\ref{th:2.1}. Thus 
\begin{equation}\label{eq:4.1}
S_n = S(\Sigma_n(A,B)) = (n-1) S(A-B) + S(A+\overline{n-1}B).
\end{equation}
Denote by $\rho^A$ the zero mean Gaussian state with
covariance matrix $A$ for any G-matrix $A$. Thus we have
\[\rho^{A+\overline{n-1}B} =\int_{\mathbb{R}^{2k}} W(\bm{\xi})
\rho^A  W(\bm{\xi})^\dag \phi( \bm{\xi})\,
\dd\bm{\xi}\]
where $\bm{\xi} = \bm{\xi}_1\oplus \bm{\xi}_2$, $W(\bm{\xi})$
is Weyl or displacement operator at $\bm{\xi}_1 + \imath
\bm{\xi}_2$ and $\phi(\bm{\xi})$ is the Gaussian density
function with mean zero and covariance matrix $(n-1)B$. By
Proposition 6.2 of \cite{MR1230389} we have 
\begin{eqnarray}
S(A+\overline{n-1}B) &=& S(\rho^{A+\overline{n-1}B})
\nonumber\\
&\le & \int S(A)\phi(\bm{\xi})\,\dd\bm{\xi}+H(\phi)
\label{eq:4.2}
\end{eqnarray}
where $H(\phi)$ is the Shannon differential entropy of the
density function $\phi$. Since (by \cite{coverthomas}) 
\begin{equation}\label{eq:4.3}
H(\phi) = k \log 2\pi e +\frac{1}{2} \log
\det(\overline{n-1} B)
\end{equation}
it follows from (\ref{eq:4.1})--(\ref{eq:4.3}) that 
\begin{eqnarray*} 
\left| \frac{S_n}{n}-\frac{n-1}{n}S(A-B)\right| &\le &
\frac{S(A)}{n} +\frac{k}{n} \log 2\pi e +\frac{1}{2n}
\log(n-1)^{2k} \det B \\
&\le & \frac{1}{n} \left[ S(A)+ k \log 2\pi e + \frac{1}{2}
\log \det B \right] +\frac{k}{n} \log(n-1).
\end{eqnarray*}
Letting $n \rightarrow \infty$ we get
\[\bar{S}(\Sigma)=S(A-B).\] 
\end{proof}
\begin{theorem}\label{th:4.3}
Let $p_1,p_2,\cdots$ be a probability distribution over
$\{1,2,3,\cdots\}$, and let $A$ and $B$ be $2k \times 2k$
symmetric real matrices satisfying the condition that $A+tB$
is a G-matrix for every $t\in[-2,2]$. Let $\Sigma$ be the
stationary G-chain defined by the infinite block Toeplitz matrix 
\[ \Sigma = \begin{bmatrix}
A & p_1B & p_2B& \cdots & \cdots \\
p_1B & A & p_1B & p_2B & \cdots \\
p_2B & p_1B & A & p_1B & \cdots \\
\vdots &\vdots &\vdots  & \ddots & \ddots
\end{bmatrix}.\]
Then the entropy rate of $\Sigma$ is given by 
\[\bar{S}(\Sigma) =\int_0^1 S(A+h(s)B)\,\dd s\]
where 
\[h(s) =2 \sum_{j=1}^\infty p_j \cos 2\pi j s, \quad
s\in[0,1].\]
\end{theorem}

\begin{proof}
We can express $\Sigma_n$ as 
\[\Sigma_n =A \otimes I_n +B \otimes T_n(\bm{p})\]
where $\Sigma_n$ is $\Sigma$ restricted to its first $n$ row
and column blocks and $T_n(\bm{p})$ is the Toeplitz matrix
given by 
\[T_n(\bm{p}) = \begin{bmatrix}
0 & p_1 & p_2 & \cdots & p_{n-1} \\
p_1 & 0 & p_1 & \cdots & p_{n-2}\\
p_2 & p_1 & 0 & \cdots & p_{n-3} \\
\vdots & \cdots & \vdots & \ddots & \vdots\\
p_{n-1} & p_{n-2} & p_{n-3} & \cdots & 0
\end{bmatrix}.\]
Let $\lambda_{n1},\lambda_{n2},\cdots,\lambda_{nn}$ be the
eigenvalues of $T_n(\bm{p})$ and let $\ket{\psi_{n1}},
\ket{\psi_{n2}}, \cdots , \ket{\psi_{nn}}$ the corresponding
eigenvectors constituting an orthonormal basis for
$\mathbb{R}^n$ so that 
\[\Sigma_n =\sum_{j=1}^n (A+\lambda_{nj}B) \otimes
\braket{\psi_{nj}}{\psi_{nj}}.\]
This shows that 
\begin{eqnarray*}
\frac{1}{n} S(\Sigma_n) &=& \frac{1}{n} \sum_{j=1}^n
S(A+\lambda_{nj}B) \\
&=& \int S(A+sB)\,\dd \mu_n(s),
\end{eqnarray*}
where $\mu_n$ is the probability measure defined by 
\[\mu_n =\frac{1}{n} \sum_{j=1}^n \delta_{\lambda_{nj}}.\]
By Kac-Murdock-Szeg\"o theorem $\mu_n$ converges weakly as
$n \rightarrow \infty$ to the distribution $Lh^{-1}$ where
$L$ denotes the Lebesgue measure in $[0,1]$ and 
\[h(s) =2 \sum_{j=1}^\infty p_j \cos 2\pi j s.\]
Note that $\|T_n(\bm{p})\| \le 2$ and the eigenvalues
$\lambda_{nj}$ lie in the interval $[-2,2]$. Furthermore,
the symplectic spectrum of $A+sB$  is a continuous function
of $s$ and hence the entropy $S(A+sB)$ is a continuous
function of $s$ in $[-2,2]$. Thus 
\begin{eqnarray*}
\lim_{n\rightarrow \infty} \frac{1}{n} S(\Sigma_n) &=&
\int_{-2}^2 S(A+sB) Lh^{-1}(\dd s) \\
&=& \int_0^1 S(A+h(s)B)\,\dd s.
\end{eqnarray*}
\end{proof}

\section{Entanglement in a stationary G-chain}

\par We have already noted in Corollary \ref{cor:2.1} that
in any exchangeable G-chain $\Sigma(A,B)$ of order $k$ all
the underlying Gaussian states $\rho(I)$ with $\# I \ge 2$
are separable. In the class of examples of stationary
G-chains in Corollary \ref{cor:3.2} it is natural to examine
the presence of entanglement. It would be particularly
interesting to construct an example of a stationary G-chain
for which the underlying Gaussian state $\rho(I\cup (I+n))$
with marginals $\rho(I)$ and $\rho(I+n)=\rho(I)$ is
entangled for arbitrarily large $n$. That would imply the
preservation of entanglement after an arbitrarily large
lapse of time. For such problems we do not have any answer.
However, we shall examine the special case of a stationary G-chain of
order $1$ where the matrices $A,\,B_1,\,B_2,\, \cdots$  in
Corollary \ref{cor:3.2} are given by 
\begin{eqnarray*}
A &=& \lambda I_2 \\
B_j &=& B = b\begin{bmatrix} 1 & 0 \\ 0 & -1 \end{bmatrix}
\quad j=1,2,\cdots,
\end{eqnarray*}
where $\lambda$ and $b$ are positive scalars with $\lambda
> \frac{1}{2}$. We start with two elementary lemmas. Let
\begin{equation}\label{eq:5.1}
 \Sigma = \begin{bmatrix}
\lambda I_2 & p_1B & p_2B& \cdots & \cdots \\
p_1B & \lambda I_2& p_1B & p_2B & \cdots \\
p_2B & p_1B & \lambda I_2 & p_1B & \cdots \\
\vdots &\vdots &\vdots  & \ddots & \ddots
\end{bmatrix}
\end{equation}
where $p_1,p_2,\cdots$ is a probability distribution over
$\{1,2,\cdots\}$ and $\lambda,\,B$ as above.
\begin{lemma}\label{lem:5.1}
The infinite block matrix $\Sigma$ in (\ref{eq:5.1}) is a
stationary G-chain of order one if $b< \frac{1}{2} \left(
\lambda^2 -\frac{1}{4} \right)^{\frac{1}{2}}$.
\end{lemma}
\begin{proof}
By Corollary \ref{cor:3.2}, $\Sigma$ is a stationary G-chain
of order one if $\lambda I_2 +t B$ is a G-matrix for all
$t\in [-2,2]$. This is fulfilled if the matrix inequalities 
\[\begin{bmatrix} 
\lambda+tb & \frac{\imath}{2} \\
-\frac{\imath}{2} & \lambda- tb
\end{bmatrix}>0, \quad  t \in [-2,2]
\]
hold. Clearly this is satisfied if $b < \frac{1}{2} \left(
\lambda^2 -\frac{1}{4} \right)^{\frac{1}{2}}$.
\end{proof}
\begin{lemma}\label{lem:5.2}
Let $\lambda > \frac{1}{2},~c>0$. Then the matrix
\[\Gamma = \begin{bmatrix}
\lambda & 0 & c & 0 \\
0 & \lambda & 0 & -c \\
c & 0 & \lambda & 0 \\
0 & -c & 0 & \lambda \end{bmatrix}\]
is the covariance matrix of an entangled $2$-mode Gaussian
state if 
\[\lambda - \frac{1}{2} < c <  \left(\lambda^2 -\frac{1}{4}
\right)^{\frac{1}{2}}.\]
\end{lemma}
\begin{proof}
By Simon's criterion \cite{PhysRevLett.84.2726}, $\Gamma$
has the required property if the following two matrix
inequalities hold:
\begin{eqnarray}
\begin{bmatrix} 
\lambda & \frac{\imath}{2} \\
-\frac{\imath}{2} & \lambda
\end{bmatrix} - c^2 \begin{bmatrix} 1 & 0 \\ 0 & -1
\end{bmatrix} \begin{bmatrix} 
\lambda & \frac{\imath}{2} \\
-\frac{\imath}{2} & \lambda
\end{bmatrix}^{-1} \begin{bmatrix} 1 & 0 \\ 0 & 1
\end{bmatrix} &>& 0, \label{eq:5.2} \\
\begin{bmatrix} 
\lambda & \frac{\imath}{2} \\
-\frac{\imath}{2} & \lambda
\end{bmatrix} - c^2 \begin{bmatrix} 1 & 0 \\ 0 & -1
\end{bmatrix} \begin{bmatrix} 
\lambda & -\frac{\imath}{2} \\
\frac{\imath}{2} & \lambda
\end{bmatrix}^{-1} \begin{bmatrix} 1 & 0 \\ 0 & 1
\end{bmatrix} &\not\ge& 0. \label{eq:5.3}
\end{eqnarray}
Simple algebra shows that (\ref{eq:5.2}) holds whenever $c^2
< \lambda^2 - \frac{1}{4}$. Inequality (\ref{eq:5.3})
reduces to 
\begin{equation}\label{eq:5.4}
\begin{bmatrix}
(1-d^2) \lambda & \frac{\imath}{2}(1+d^2) \\
-\frac{\imath}{2}(1+d^2) & (1-d^2) \lambda
\end{bmatrix} \not\ge 0
\end{equation}
where 
\begin{equation}\label{eq:5.5}
d = \frac{c}{\sqrt{\lambda^2 -\frac{1}{4}}}.
\end{equation} 
We now choose $d<1$ such that the determinant of the matrix
on the left hand side of (\ref{eq:5.4}) is strictly
negative. This finally leads to the inequality $c>\lambda
-\frac{1}{2}$. Thus the inequality $\lambda - \frac{1}{2} <
c < \sqrt{\lambda^2 -\frac{1}{4}}$ is sufficient to ensure
entanglement. 
\end{proof}
\begin{prop}\label{prop:5.3}
Let $\frac{1}{2} < \lambda < \frac{5}{6}$, $\lambda -
\frac{1}{2} < b < \frac{1}{2} \sqrt{\lambda^2
-\frac{1}{4}}$. Suppose $p_j b > \lambda - \frac{1}{2}$ for
some $j$. Then the $2$-mode Gaussian state $\rho(\{1,j\})$
determined by the stationary G-chain $\Sigma$ defined by
(\ref{eq:5.1}) is entangled.
\end{prop}
\begin{proof}
Since $\frac{1}{2} < \lambda < \frac{5}{6}$, the interval
$\left( \lambda -\frac{1}{2} , \frac{1}{2}\left(\lambda^2
-\frac{1}{4} \right)^{\frac{1}{2}} \right)$ is open and
nonempty. Thus it
is possible to choose $b$ in this interval. By Lemma
\ref{lem:5.1} the block matrix $\Sigma$ in (\ref{eq:5.1}) is
a stationary G-chain of order one. The covariance matrix of
the $2$-mode state $\rho(\{1,j\})$ defined by the G-chain
$\Sigma$ is equal to 
\[ \begin{bmatrix}
\lambda & 0 & p_jb & 0 \\
0 & \lambda & 0 & -p_jb \\
p_jb & 0 & \lambda & 0 \\
0 & -p_jb & 0 & \lambda \end{bmatrix}\]
where by hypothesis 
\[\lambda -\frac{1}{2} < p_j b < \left( \lambda^2
-\frac{1}{4}\right)^{\frac{1}{2}}.\]
By Lemma \ref{lem:5.2} it follows that $\rho(\{1,j\})$ is
entangled. 
\end{proof}
\begin{rem}
\emph{It follows from Proposition \ref{prop:5.3} that for
$\frac{1}{2} < \lambda < \frac{5}{6}, ~ \lambda-\frac{1}{2} <
b < \frac{1}{2} \left(\lambda^2 -\frac{1}{4}
\right)^{\frac{1}{2}}$ the matrix 
\[\Sigma = \left[
\begin{array}{rrrrrrrrrrr}
\lambda & 0 &  b & 0 & 0 & 0 & \cdot & \cdot & \cdot &
\cdots \cdots\\
0 & \lambda & 0 & -b & 0 & 0 & \cdot & \cdot & \cdot &
\cdots\cdots\\
b & 0 & \lambda & 0 & b & 0 & 0 & 0 & \cdot & \cdots \cdots
\\
0 & -b & 0 & \lambda & 0 & -b & 0 & 0 & \cdot & \cdots
\cdots\\
0 & 0 & b & 0 & \lambda & 0 & b & 0 & \cdot & \cdots
\cdots\\
0 & 0 & 0 & -b & 0 & \lambda & 0 & -b & \cdot & \cdots
\cdots \\
\vdots & \vdots &\vdots &\vdots &\vdots &\vdots &\vdots
&\vdots &\ddots &\cdots\cdots \\
\vdots & \vdots &\vdots &\vdots &\vdots &\vdots &\vdots
&\vdots &\vdots &\cdots\cdots 
\end{array} \right]\]
(which is $2 \times 2$ block tridiagonal)  is a stationary
G-chain in which $\rho(\{1,2\}) = \rho(\{n , n+1\}), ~n
=2,3,\cdots$ is entangled. In this example $\rho(\{n, n
+j\})$, $j\ge2$ is a product state with covariance matrix 
\[\begin{bmatrix} \lambda I_2 & 0 \\ 0 & \lambda I_2
\end{bmatrix}.\] }
\end{rem}
\section{Conclusion}
\par An exchangeable chain of mean zero finite mode Gaussian
states is completely determined by two matrices $A,\, B$
such that $A-B$ is the covariance matrix of a Gaussian state
and $B$ is a nonnegative definite matrix. Its asymptotic
entropy rate is equal to $S(\rho^{A-B})$, the von Neumann
entropy of the mean zero Gaussian state $\rho^{A-B}$ with
covariance matrix $A-B$. 
\par A class of examples of stationary Gaussian chains is
constructed and their asymptotic entropy rates are
evaluated. An example of a stationary entangled chain is
presented. 
\bibliographystyle{alpha}
\bibliography{biblio}

\end{document}